\documentclass[12pt,reqno]{amsart}
\usepackage[headings]{fullpage}
\usepackage{amssymb,amsmath,amscd,bbm,tikz}
\usepackage{graphicx}
\usepackage{texdraw}
\usepackage{braket}
\usepackage{pb-diagram}
\usepackage{url}
\usepackage[T1]{fontenc}
\usepackage{slashed}    
\usepackage{tikz-cd}
\usetikzlibrary{decorations.markings}
\usetikzlibrary{decorations.pathmorphing}
\tikzset{snake it/.style={decorate, decoration=snake}}
\usetikzlibrary{patterns}
\usepackage{stmaryrd}
\usepackage{tabularray}

\usepackage{todonotes}
\usepackage{enumerate}
\usepackage{mathtools}

\usepackage{arydshln}

\usepackage{verbatim} 
\usepackage{xcolor}

\usepackage{hyperref}
\usepackage{cleveref}

\newtheorem{theorem}{Theorem}[section]
\theoremstyle{definition}
\newtheorem{prop}[theorem]{Proposition}
\newtheorem{lemma}[theorem]{Lemma}
\newtheorem{definition}[theorem]{Definition}
\newtheorem{remark}[theorem]{Remark}

\newtheorem{question}[theorem]{Question}

\def\BN{\mathbb N}

\def\BP{\mathbb P}

\def\calR{\mathcal R}
\def\calW{\mathcal W}

\def\calK{\mathcal K}

\def\calC{\mathcal C}

\def\calP{\mathcal P}
\def\calS{\mathcal S}
\def\calB{\mathcal B}

\def\calM{\mathcal M}
\def\calL{\mathcal L}

\def\calN{\mathcal N}

\def\ri{\mathrm{i}}

\newcommand\Res[1]{\,\underset{#1}{\mathrm{Res}\,}}

\def\Re{\mathrm{Re}}
\def\re{\mathrm{e}}
\def\Im{\mathrm{Im}}

\def\be{\begin{equation}}
\def\ee{\end{equation}}

\def\bb{\mathsf{b}}
\def\cb{c_{\mathsf{b}}}

\usepackage{accents}

\newcommand{\BR}{\mathbb{R}}
\newcommand{\BQ}{\mathbb{Q}}
\newcommand{\BC}{\mathbb{C}}
\newcommand{\BH}{\mathbb{H}}
\newcommand{\BZ}{\mathbb{Z}}

\makeatletter

\renewcommand\thepart{\@Roman\c@part}%
\renewcommand\part{%
   \if@noskipsec \leavevmode \fi
   \par
   \addvspace{6.7ex}%
   \@afterindentfalse
   \secdef\@part\@spart}
\def\@part[#1]#2{%
    \ifnum \c@secnumdepth >\m@ne
      \refstepcounter{part}%
      \addcontentsline{toc}{part}{Part~\thepart.\ #1}%
    \else
      \addcontentsline{toc}{part}{#1}%
    \fi
    {\parindent \z@ \raggedright
     \interlinepenalty \@M
     \normalfont
     \ifnum \c@secnumdepth >\m@ne
       \centering\large\scshape \partname~\thepart.%
       \hspace{1ex}%
     \fi%
     \large\scshape #2%
     \markboth{}{}\par}%
    \nobreak
    \vskip 4.7ex
    \@afterheading}
  \def\@spart#1{
  \refstepcounter{part}%
  \addcontentsline{toc}{part}{#1}%
    {\parindent \z@ \raggedright
     \interlinepenalty \@M
     \normalfont
     \centering\large\scshape #1\par}%
     \nobreak
     \vskip 4.7ex
     \@afterheading}
\renewcommand*\l@part[2]{%
  \ifnum \c@tocdepth >-2\relax
    \addpenalty\@secpenalty
    \addvspace{0.75em \@plus\p@}%
    \begingroup
      \parindent \z@ \rightskip \@pnumwidth
      \parfillskip -\@pnumwidth
      {\leavevmode
       \normalsize \bfseries #1\hfil \hb@xt@\@pnumwidth{\hss #2}}\par
       \nobreak
       \if@compatibility
         \global\@nobreaktrue
         \everypar{\global\@nobreakfalse\everypar{}}%
      \fi
    \endgroup
  \fi}

\def\l@subsection{\@tocline{2}{0pt}{2pc}{6pc}{}}
\makeatother

\hypersetup{colorlinks=true,linkcolor=blue,citecolor=teal,filecolor=magenta,urlcolor=cyan}

\usepackage[backend=bibtex,style=alphabetic,sorting=nyt,isbn=false,url=false,doi=false,maxalphanames=4,minalphanames=4,mincitenames=4,maxcitenames=5,minnames=4,maxnames=10,giveninits=true,maxbibnames=99]{biblatex}
\renewbibmacro{in:}{}

\numberwithin{equation}{section}

\usepackage{microtype}
\begin{document}
\title[Quantum modularity in refined topological recursion]{
 Quantum modularity in refined topological recursion}
\author{Veronica Fantini}
\address{Laboratoire Math\'{e}matique d'Orsay, Universit\'{e} Paris Saclay \\ 
          \newline
          {\tt \url{https://sites.google.com/view/vfantini/home-page}}}
\email{veronica.fantini@universite-paris-saclay.fr}
\author{Kento Osuga}
\address{Kobayashi--Maskawa Institute \& Graduate School of Mathematics, Nagoya University
          \newline
         {\tt \url{https://sites.google.com/view/kentoosuga/}}}
\email{osuga@math.nagoya-u.ac.jp}

\date{\today }

\begin{abstract}
  In this paper we study the Borel summability and quantum modularity of the free energy computed by refined topological recursion on the Weber-type refined spectral curve. On the one hand, we verify expected properties, namely, Stokes rays come in a pair and the Borel--Laplace sum of the free energy coincides with the Barnes double gamma function. Somewhat unexpectedly, on the other hand, we prove that the generating function of the Stokes constants along each ray is a Jacobi quantum modular function. Our proof suggests that quantum modularity may be a general feature of refined topological recursion, applicable beyond the Weber-type refined spectral curve.
\end{abstract}

\maketitle

\tableofcontents

\newpage

\section{Introduction}
 In this paper, we will uncover an intriguing appearance of quantum modularity in refined topological recursion on the Weber-type refined spectral curve which is of genus zero.

\subsection{Backgrounds}

\subsubsection{Modularity in topological strings}

In the context of topological strings, two types of modularity have been commonly discussed; with respect to a coordinate $\tau$ on the corresponding complex moduli space, or with respect to the string coupling $\hbar$ (often denoted by $g_s$). The former is expected as a consequence of the holomorphic anomaly~\cite{ABK}, and indeed, the genus-$g$ free energy of several Calabi--Yau threefolds are proven to be quasi modular for each $g\geq2$ \cite{CI,Guo:2026xpg,Douaud-KP26}. For local $\BP^2$, it is also known that the genus-$g$ free energy in the Nekrasov--Shatashvili limit is a quasi modular function ~\cite{Bousseau:2020ckw}. 

On the other hand, modularity with respect to the string coupling $\hbar$ arises when one considers the resurgent structure of the formal asymptotic series of certain invariants. For instance, such a modularity is observed in the framework of the topological string/spectral theory correspondence \cite{GHM,CHM}. More concretely, building on the results of \cite{rella22}, the resurgent structure of so-called fermionic traces for weighted projective spaces is investigated in \cite{FR1phys,FR-pmn} and shown to give rise to quantum modular functions.

\subsubsection{Holomorphic quantum modular forms}

Quantum modular forms were introduced by Zagier in~\cite{zagier_modular} to describe modular properties of quantum invariants of knots and $3$-manifolds, following the pioneering work of Lawrence and Zagier~\cite{qCS1}. Later in 2020, a class of quantum modular forms, called \emph{holomorphic quantum modular forms}, was introduced by Zagier~\cite[min 21:45]{zagier-talk}.  
Since then, their study has become active in number theory (e.g. \cite{MS-R,FR1maths,gene}), quantum topology\footnote{In this context, the modular parameter is related to the level $k$ of the Chern--Simons action.} (e.g. \cite{DG-qmod,Wheeler-3manifold,yuya,GZ1,GZ2} and references therein), and more. 

As holomorphic quantum modular forms are relatively new objects to study, it is meaningful to find concrete examples in different areas of mathematics and physics. In this regard, our main result serves as a new observation of quantum modularity. It is worth emphasising that this modularity is not with respect to the string coupling $\hbar$ nor the complex moduli parameter $\tau$, but with respect to the refinement parameter $\beta$ in refined topological recursion while the refined spectral curve is of genus zero, and moreover, while $\hbar$ is present and finite.

\subsubsection{Refined topological recursion}

Topological recursion due to Eynard and Orantin \cite{EO07} is a universal recursive mechanism that governs various invariants from algebraic geometry to mathematical physics. One of the most notable outcomes of topological recursion is the remodeling theorem \cite{BKMP07,FLZ16} which can be thought of as an all-genus mirror symmetry for toric Calabi--Yau threefolds. It has been shown that topological recursion also plays important roles in Hurwitz theory \cite{BM07,ACEH18}, integrable hierarchies \cite{BDKS20,ABDKS24}, $\calW$-algebras \cite{KS17,ABCO17,BBCCN18}, and more.

The idea of \emph{refined topological recursion} is discussed in \cite{CE06,C10,BMS10} in the context of $\beta$-deformed matrix models, and its mathematical description was recently formulated in \cite{KO22,Osu23-1} (for degree-two curves). This has triggered a new research direction connecting with geometry and combinatorics in the non-orientable setting, e.g. $b$-Hurwitz numbers \cite{CDO24-1,CDO24-2}, moduli spaces of bordered hyperbolic Klein surfaces \cite{GGO25}, and moduli spaces of metric M\"{o}bius graphs \cite{CGGO26}. Refined topological recursion yet enjoys another property in terms of refined BPS invariants \cite{KO23,KW24}.

\subsection{Main results}

We focus on the so-called Weber-type refined spectral curve $\calS^\mu$. It consists of a degree-two genus-zero algebraic curve $\calC$ and a certain choice of multidifferentials $\{\omega_{0,1}, \omega_{0,2},\omega_{1/2,1}\}$. Refined topological recursion takes $\calS^\mu$ as initial data, and it returns a sequence of multidifferentials $\{\omega_{g,n}\}_{g,n}$ labeled by $g\in\frac12\BZ_{\geq0}$ and $n\in\BZ_{\geq0}$. A key result for us is proven in \cite{KO23} that, for $g>1$, the genus-$g$ free energy $F_g\coloneqq\omega_{g,0}$ of $\calS^\mu$ is derived explicitly in terms of the double Bernoulli polynomial $B_{2,k}(x|a_1,a_2)$ as below:
\begin{equation}
  F_g=\frac{(-1)^{2g-2} B_{2,2g}\left(\frac{\mathfrak{b}}{2}(1-\mu)|\beta^{1/2},-\beta^{-1/2}\right)}{2g(2g-1)(2g-2)} t^{2-2g}\label{intro:RTRF_g}
\end{equation}
where $t,\mu$ are related to certain residues of $\omega_{0,1},\omega_{1/2,1}$, the other two parameters $\mathfrak b$ and $\beta$ are related to each other by $\mathfrak b = \beta^{1/2}-\beta^{-1/2}$, and $F_g$ with $g\leq1$ are also determined separately. Assembling $F_g$ for $g>1$, we define the \emph{stable free energy} $F^{\rm st}$ as a formal series in $\hbar$:
\begin{equation}
  F^{\rm st}\coloneqq\sum_{g\in\frac12\BZ_{>2}}\hbar^{2g-2}\,F_g.
\end{equation}

So how can modularity possibly appear in refined topological recursion? Since a structural relation (called modular resurgence) between resurgence and modularity in $\hbar$ is recently conjectured in \cite{FR1maths}, a constructive approach is to consider the Borel-summability of $F^{\rm st}$. By doing so, we find that non-pertubative corrections to $F^{\rm st}$ are quantum modular functions, but somewhat surprisingly, with respect to $\beta$, not $\hbar$ --- a modularity in $\beta$ is implicitly hinted also in \cite{Zam05}, though not quantum modularity.
 
Let us first summarise our result on the summability of $F^{\rm st}$. In short, a formal series $f$ is called \emph{Borel-summable} if $s^\theta[f]$, the Borel--Laplace transform of $f$ in the direction of $\theta$, satisfy certain analytic conditions. Note that, although the summability of $F^{\rm st}$ is expected from the $\beta$-deformed matrix model perspective, no proof has been given in line with refined topological recursion, hence we prove it for completeness. In particular, in order to remove ambiguities, we set the branch of $\beta^{1/2}$ along the negative real axis which ensures that ${\rm Re}(\beta^{\pm1/2})>0$:

\begin{theorem}\label{thm-1-intro}
The stable free energy $F^{\rm st}\in\BQ[\mu,\beta^{\pm1/2}][\![\hbar/t]\!]$ is Borel-summable, and for ${\rm Re}(\hbar/t)>0$, $s^0[F^{\rm st}]$ is given by the Barnes double gamma function:
\begin{equation}\label{intro:borel_sum_Fstable}
\frac{F_0}{\hbar^2}+\frac{F_{\frac12}}{\hbar}+F_1+s^0[F^{\rm st}](\hbar;\beta)=\log\Gamma_2\left(\frac{t}{\hbar}+\frac{\mathfrak{b}}{2}(1-\mu)\bigg| \beta^{1/2},-\beta^{-1/2}\right)\,.
\end{equation}
\end{theorem}

It can be shown that the Borel transform of $F^{\rm st}$ has singularities at $2\pi \ri t \beta^{\pm1/2} k$ for $k\in\BZ_{\neq 0}$. As a consequence, $s^{\theta}[F^{\rm st}]$ is discontinuous as $\theta$ varies across the direction of $\theta_{\pm}\coloneqq\arg\ri t\beta^{\pm1/2}$ and also $\theta_\pm+\pi$. We, therefore, define ${\rm disc}_\theta[F^{\rm st}]\coloneqq s^{\theta+\varepsilon}[F^{\rm st}]-s^{\theta-\varepsilon}[F^{\rm st}]$ for a sufficiently small $\varepsilon\in\mathbb{R}_{>0}$ which is nontrivial only when $\theta\in\{\theta_\pm,\theta_\pm+\pi\}$ and describes the non-pertubative corrections. Note that when $\beta=1$, we have $\theta_+=\theta_-$ so that two rays coincide. In other words, the refinement parameter $\beta$ splits the unrefined Borel singularity into two. This phenomenon has already been observed in \cite{AMP23}, which reduces to the case with $\mu=0$.

\bigskip 

We will next summarise our observation about quantum modularity, and to this end, let us clarify our setting. First, we choose $\beta\in\BH$ which we think of as the corresponding modular parameter. This means that such a modularity becomes invisible once we set $\beta$ to a specific value (e.g. $\beta=1$ which corresponds to the Eynard--Orantin recursion). 
Second, we regard $(\beta,\epsilon_2)$ as independent variables instead of $(\beta,\hbar)$ where $\epsilon_2\coloneqq-\beta^{-1/2}\hbar$ and $\epsilon_1\coloneqq\beta^{1/2}\hbar$ --- one can equivalently choose $(-\beta^{-1},\epsilon_1)$ as independent variables, if preferred. We also note that this is equivalent to rescaling $\omega_{g,n}$ by $\beta^{g/2}$, and $\omega_{g,n}$ often behaves better after such a simple rescaling as observed in \cite{CDO24-2,CGGO26}.

Our primary interest is in so-called \emph{Jacobi quantum modular forms}. Loosely speaking, a Jacobi quantum modular form is a function $f:\BC\times\BH\to\BC$ that obeys specific transformation laws under the modular group $\Gamma\subseteq \mathrm{SL}_{2}(\BZ)$. An important difference from standard modular forms is that it carries two variables, the one in $\BC$ is called the elliptic variable and the other in $\BH$ is the modular parameter. Our main result shows that each nontrivial discontinuity turns out to be a Jacobi quantum modular function. To be more precise, we have:

\begin{theorem}\label{thm-2-intro}
Set $\mu\in\BZ$. The discontinuities ${\rm disc}_\theta[F^{\rm st}]$ in the direction of $\theta\in\{\theta_\pm\,,\theta_\pm+\pi\}$ are all Jacobi quantum modular functions with respect to the elliptic variable $t/\epsilon_2$ and the modular parameter $\beta$.
\end{theorem}

The strategy of the proof can be outlined as follows. First, we explicitly obtain that each discontinuity ${\rm disc}_\theta[F^{\rm st}]$ admits a $q$-series where $q=\re^{2\pi\ri\beta}$ (or $\tilde q$-series where $\tilde q=\re^{-2\pi\ri/\beta}$) which guarantees the invariance under the $T$-action $\beta\mapsto\beta+1$ (pre-composed by the $S$-action $\beta\mapsto-\beta^{-1}$ for the $\tilde{q}$-series). Next, we find that the $S$-action $\beta\mapsto-\beta^{-1}$ brings the discontinuity along $\theta_+$ to the minus of that along $\theta_-$. In other words, the $S$-cocycle of the discontinuity along $\theta_+$ is the sum of the two discontinuities along $\theta_+$ and $\theta_-$. Finally, we show that the sum has a good analytic behaviour when $\mu\in\BZ$, in particular, it is partly written in terms of the Faddeev's quantum dilogarithm $\Phi_{\beta^{1/2}}$.

Our proof exhibits a somewhat structural property, i.e. existence of a $q$-series and the $S$-cocycle as the sum of the two discontinuities split due to $\beta$. Therefore, it is natural to expect that the quantum modularity with respect to the refinement parameter $\beta$ may hold beyond the Weber-type refined spectral curve. The explicit formula for the resolved conifold (see Section \ref{sec:conclusion}) and the discussions about conifold frames in \cite{AMP23} support this expectation, and we hope to return to this idea with more nontrivial examples in the near future.

\bigskip

This paper is organised as follows. In Section \ref{sec:RTR}, we briefly introduce refined topological recursion and discuss necessary properties. We then review basic notions of resurgence theory in Section \ref{sec:resurgence} and prove Theorem~\ref{thm-1-intro}. We turn to quantum modularity in Section \ref{sec:modularity}, show our main result Theorem~\ref{thm-2-intro}, and conclude with comments for future work.

\subsection*{Acknowledgements}
The authors would like to thank T. Bridgeland, A. Brini, A. Grassi, K. Iwaki, and O. Kidwai for helpful comments and valuable discussions. We also thank Nagoya University for their hospitality at which this project started.

K. O. acknowledges the support by JSPS KAKENHI Grant-in-Aid for Early-Career Scientists (23K12968, 26K16980), and in part by 24K00525. K. O. also acknowledges the support from the Kobayashi–Maskawa Institute (KMI) for the Origin of Particles and the Universe at Nagoya University.

\bigskip

\section{Refined topological recursion}\label{sec:RTR}

Let us first introduce some aspects of recent progress on refined topological recursion. Our purposes are not to give a thorough introduction but rather collect results necessary for our main claim. Thus, we will be very brief and refer the readers to \cite{KO22,Osu23-1,Osu23-2} for more details and also \cite{CE06} for initial attempts.

\subsection{Refined spectral curves}

Although refined topological recursion can be considered for any degree-two curves, our primary interests are in the model built on the so-called Weber curve, one-parameter family of the following algebraic curves:
\be\label{eq:weber}
\calC=\{x,y\in\BC,\;\;\;t\in\BC^*\,|\,y^2-\tfrac{1}{4}x^2+t=0\}
\ee
Strictly speaking, we consider its compactification $\overline{\calC}$ which, for fixed $t$, is isomorphic to $\BP^1$. For brevity of notation, however, we drop the bar and simply denote by ${\calC}$ its compactification as this difference plays no role in the present paper.

The parameter $t$ can be geometrically realised as follows. Let $\calP\subset{\calC}$ be the set of poles of $ydx$ which consists of two points. Then, it is easy to show that 
\begin{equation}\Res{p\in\calP}\, ydx =\pm t \end{equation}
We write $\calP=\{p_+,p_-\}$ where the subscript denotes the sign of the above residue. Let us also set $\calR\coloneqq\{p\in{\calC}\,|\,dx(p)=0\}$ the set of ramification points of $x:\calC\to\BP^1$.

We next consider a refined spectral curve which is a set of initial data required to apply refined topological recursion. Since we only focus on a specific model, the definition can be simplified as below --- it is not written exactly in the same manner as \cite{KO22,Osu23-1} but it can be shown to be equivalent for the Weber curve:
\begin{definition}
The Weber-type \emph{refined spectral curve} $\calS^\mu=({\calC},\omega_{0,1},\omega_{0,2},\omega_{\frac12,1})$ consists of the following data:
\begin{itemize}
\item ${\calC}$ as above,
\item $\omega_{0,1}\coloneqq ydx$
\item $\omega_{0,2}$ is the fundamental bi-differential on ${\calC}^2$
\item $\omega_{\frac12,1}$ is a unique differential on ${\calC}$ whose only poles are residues\footnote{The sign convention for $\mu$ is the same as \cite{CDO24-2,CGGO26} but the opposite of \cite{KO23}.} determined as below:
\begin{equation}\Res{p\in\calR}\omega_{\frac12,1}=-\frac{\mathfrak{b}}{2},\quad \Res{p=p_\pm}\omega_{\frac12,1}=\frac{\mathfrak{b}(1\mp\mu)}{2}\end{equation}
where $\mathfrak{b}\in\BC$ is called the \emph{refinement parameter} and $\mu\in\BC$ is another parameter. 
\end{itemize}
\end{definition}

\begin{remark}
We will consider three parameters $\mathfrak{b},\beta,b$ which are mutually related by
\begin{equation}\mathfrak{b}=\beta^{\frac12}-\beta^{-\frac12},\quad b =\beta^{-1}-1\,,\end{equation}
As shown in \cite{CDO24-1,CDO24-2}, the central charge $c$ of the associated Virasoro algebra is given by $c=1-6\mathfrak{b}^2$. Also, $\beta$ coincides with that of the $\beta$-deformed Hermitian matrix models, and it is related to refined topological string. Furthermore, yet another parameter $b$ admits a combinatorial interpretation via measure of non-orientability \cite{CD20,CGGO26}. 
\end{remark}

\subsection{Free energies}
Starting with the refined spectral curve $\calS^\mu$, \emph{refined topological recursion} is a unique construction of a sequence of symmetric multi-differentials $\omega_{g,n}$, i.e. meromorphic sections of $K_{{\calC}}^{\boxtimes n}$, labeled by $n\in\BZ_{\geq1}$ and $g\in\frac12\BZ_{\geq0}$ for $2g-2+n>0$. The defining formula for $\omega_{g,n}$ has a pants-decomposition structure, and it is recursive in $2g-2+n$. Although the recursive formula itself exhibits an important structure, we omit to write it down explicitly as it does not play any role here. See \cite{KO22,Osu23-1,Osu23-2} for a precise definition.

Suppose the sequence $\omega_{g,n}$ on the refined spectral curve $\calS^\mu$ is obtained by refined topological recursion. One can then define the \emph{genus $g$ free energy} $F_g$ for $g\in\frac12\BZ_{\geq0}$ which can be thought of as $\omega_{g,0}$. For the Weber-type refined spectral curve, its defining equation for $g\geq\frac32$ is given as below
\begin{equation}F_g\coloneqq\frac{1}{2-2g}\sum_{r\in\calR}\underset{p=r}{\text{Res}}\,\,\omega_{g,1}(p)\int^p\omega_{0,1}.
\end{equation}
The constant of integration does not matter because $\omega_{g,n}$ for $2g-2+n>0$ have no residues as proved in \cite{KO22}. The definitions of $F_g$ with $g\in\{0,\frac12,1\}$ are given based on the so-called variational formula \cite{Osu23-2}.

A key result proven in \cite{KO23} is that  $F_g$ enjoys a closed formula in terms of double Bernoulli polynomials $ B_{2,k}(x|a_1,a_2)$ where they are defined by the following expansion
\begin{equation}
\frac{z^2\re^{x z}}{(\re^{a_1 z}-1)(\re^{a_2 z}-1)}\eqqcolon\sum_{k\in\BZ_{\geq0}} B_{2,k}(x|a_1,a_2)\frac{z^k}{k!}.\label{def of doubleB}
\end{equation}
Throughout the paper, we set $a_1=\beta^{\frac12}$ and $a_2=-\beta^{-\frac12}$. Thus, for brevity of notation, we will write as below 
\begin{equation}\label{eq:sf_Bernoulli}
  \mathsf  B_{2,k}(x)\coloneqq  B_{2,k}(x|a_1,a_2)
\end{equation}

\begin{theorem}[\cite{KO23}]
The genus $g$ free energy $F_g$ of the Weber-type refined spectral curve $\calS^\mu$ is given as below:
\begin{equation}F_0=\frac12t^2\log t-\frac{3}{4}t^2,\quad F_\frac12=-\frac{\mathfrak{b}}{2}\mu\, t\log t+\frac12\mathfrak{b}\mu t,\quad F_1=-\frac{2+(1-3\mu^2)\mathfrak{b}^2}{24}\log t,\end{equation}
and for all $g\geq\frac32$,
\begin{equation}
F_g=\frac{(-1)^{2g-2}\mathsf B_{2,2g}\left(\frac{\mathfrak{b}}{2}(1-\mu)\right)}{2g(2g-1)(2g-2)} t^{2-2g}\label{RTRF_g}
\end{equation}
\end{theorem}

Here and hereafter, $F_g$ denotes the above specific one, and any other free energies would be denoted with other indices or labels. 
For $g\in\{0,\frac12,1\}$, $F_g$ has a polynomial ambiguity of degree $2-2g$ as stated in \cite{Osu23-2}. We further utilise this ambiguity so that we can package them into the following simple form which we call the \emph{unstable part} of the free energy:
\begin{align}F^{\rm unst}(\hbar/t;\beta)\coloneqq&\frac{F_0}{\hbar^2}+\frac{F_\frac12}{\hbar}+F_1+\tfrac{1}{2} {\sf B}_{2,2}(\tfrac{t}{\hbar}+\tfrac{\mathfrak{b}}{2}(1-\mu))\log\hbar\nonumber\\
=&-\tfrac{1}{2} {\sf B}_{2,2}(\tfrac{t}{\hbar}+\tfrac{\mathfrak{b}}{2}(1-\mu))\log\tfrac{t}{\hbar}+\tfrac{3}{4}\big(\tfrac{t}{\hbar}\big)^2\mathsf B_{2,0}(\tfrac{t}{\hbar}+\tfrac{\mathfrak{b}}{2}(1-\mu))+\tfrac{t}{\hbar}{\sf B}_{2,1}(\tfrac{\mathfrak{b}}{2}(1-\mu)).\label{F^unst}
\end{align}
Note that the last term with $\log\hbar$ is added so that $F^{\rm unst}$ only depends on $\hbar/t$ and $\beta$. 

In this letter, our primary interests are in stable terms. More concretely, in contrast to $F^{\rm unst}$, we define the \emph{stable part} of the free energy $F^{\rm st}$ as below
\begin{equation}\label{eq:F_stable}
F^{\rm st}(\hbar/t;\beta)\coloneqq\sum_{g\geq\frac32}F_g\,\hbar^{2g-2}.
\end{equation}
Since $t$ only appears as $t^{2-2g}$ in each $F_g$ for $g\geq\frac32$, it follows that $F^{\rm st}\in\BQ[\mu,\beta^{\pm\frac12}][\![\hbar/t]\!]$. It turns out that $F^{\rm st}$ is a divergent series in $\hbar/t$ which leads us to consider its summability and discontinuities in Section \ref{sec:resurgence}. 

\begin{remark}\label{rem:t-dependence}
Since $F^{\rm st}$ and $F^{\rm unst}$ both depend only on $\hbar/t$ and $\beta$, one can set $t=1$ without loss of generality, and the $t$-dependence can be always restored by rescaling $\hbar$. Note that the $t$-dependence sometimes becomes useful and meaningful (e.g. \cite{KO23,CGGO26}).
\end{remark}

\subsection{Applications}

The parameter $\mu$ in the refined spectral curve $\calS^\mu$ was first introduced in \cite{KO22} merely as extra degrees of freedom hidden in refined topological recursion. It is then noticed in \cite{KO23} that it can be interpreted as a \emph{quantum BPS central charge} (different from the Virasoro central charge) in line with refined BPS structures \cite{BBS19}. It turns out that by setting $\mu$ to certain values, $F_g$ admits different enumerative geometric interpretations.

\subsubsection{$\mu=0$: conifold limit of refined topological string/refined Gromov--Witten}\label{ssub:mu_0}

Let $X$ be a Calabi-Yau threefold whose $\overline{\calM}_g(X,\beta)$ is actually compact, and consider a Calabi--Yau fivefold $X\times\BC^2$. There is a natural toric action $T$ on $\BC^2$ whose coordinates are denoted by $q_i=\re^{\epsilon_i}$ for $i\in\{1,2\}$. The $T$-equivariant Gromov--Witten theory of the fivefold $X\times\BC^2$ is a mathematical model of refined topological string theory \cite{IKV09,BS24}, and the equivariant parameters $\epsilon_1,\epsilon_2$ are related to $\hbar,\beta$ by
\begin{equation}\label{eq:epsilon}
  \epsilon_1=\beta^{1/2}\hbar,\quad \epsilon_2=-\beta^{-1/2}\hbar.
\end{equation}
See \cite{HK10,KW10,AMP23} and references therein for several computations.

Let us denote by $F^{\rm TS}_g$ the genus $g$ refined Gromov--Witten invariant, equivalently, genus $g$ refined topological string free energy. In the so-called conifold limit which is characterised by a certain coordinate $t_c\to0$ in the corresponding K\"{a}hler moduli space, $F^{\rm TS}_g$ for $g\in\BZ_{\geq2}$ is generically expected to behave as 
\begin{equation}F^{\rm TS}_g=\frac{c_g(\beta)}{t_c^{2g-2}}+O(t_c)\end{equation}
where $c_g(\beta)$ is expressed in terms of Bernoulli numbers $B_k$ by
\begin{equation}c_g(\beta)=-(2g-3)!\sum_{m=0}^g(2^{1-2m}-1)(2^{1-2g+2m}-1)\frac{B_{2m}}{2m!}\frac{B_{2g-2m}}{(2g-2m)!}\beta^{2m-g}.\label{cg(beta)}\end{equation}

Notice that $F^{\rm TS}_g$ is labelled by integer $g$, not half-integer unlike $F_g$ of refined topological recursion. Despite these differences, their precise relation is summarised as below --- this is perhaps expected from the matrix model perspective (e.g. \cite{IO10} discusses a necessity of a shift between $\beta$-deformed matrix model partition functions and Nekrasov partition functions.), but not clearly and explicitly stated anywhere to the best of author's knowledge, so we will give a self-contained proof:

\begin{prop}\label{prop:RTR-AMP}
For all $g\in\BZ_{\geq2}$, we have:
\begin{equation}F_{g-\frac12}\big|_{\mu=0}=0,\quad F_g\big|_{\mu=0} = \frac{c_g(\beta)}{t^{2g-2}}.\end{equation}
\end{prop}
\begin{proof}
The first equation is implied by the following identity when $\mu=0$
\begin{equation}B_{2,k}(x|a_1,a_2)=(-1)^k B_{2,k}(a_1+a_2-x|a_1,a_2).\label{Bsign}\end{equation}
Let us next recall the definitions of Bernoulli numbers $B_k$ and Bernoulli polynomials $B_k(a)$
\begin{equation}\frac{z}{\re^z-1}=\sum_{k\geq0}B_k\frac{z^k}{k!},\qquad \frac{z \re^{x z}}{\re^z-1}=\sum_{k\geq0}B_k(x)\frac{z^k}{k!}\end{equation}
They satisfy several identities among themselves, but the most useful ones for us are
\begin{equation}
B_k\big(\tfrac12\big)=(2^{1-k}-1)B_k,\quad \forall\, \text{ odd }k\;\;\;\;B_k=0.\label{B identity2}
\end{equation}
which implies that $c_g(\beta)$ can be written as
\begin{equation}c_g(\beta)=-(2g-3)!\sum_{m=0}^g\frac{B_{2m}\left(\frac12\right)}{(2m)!}\frac{B_{2g-2m}\left(\frac12\right)}{(2g-2m)!}\beta^{2m-g}\end{equation}
Now, from the definition of double Bernoulli polynomials \eqref{def of doubleB}, we find
\begin{align}
\sum_{k\in\BZ_{\geq0}}\mathsf B_{2,k}\big(\tfrac{\mathfrak{b}}{2}\,\big)\frac{z^k}{k!}=&\frac{z^2\re^{\frac12\beta^{\frac12}z}\re^{-\frac12\beta^{-\frac12}z}}{(\re^{\beta^{\frac12}z}-1)(\re^{-\beta^{-\frac12}z}-1)}\nonumber\\
=&-\sum_{k_1,k_2\geq0}B_{k_1}\big(\tfrac12\big)\frac{z^{k_1}\beta^{\frac{k_1}{2}}}{k_1!}B_{k_2}\big(\tfrac12\big)\frac{(-1)^{k_2}z^{k_2}\beta^{-\frac{k_2}{2}}}{k_2!}\nonumber\\
=&-\sum_{k\geq0}z^{2k}\sum_{m=0}^{k}\frac{B_{2m}\left(\frac12\right)}{(2m)!}\frac{B_{2k-2m}\left(\frac12\right)}{(2k-2m)!}\beta^{2m-k}
\end{align}
where at the first equality we substitute the definition of Bernoulli polynomials $B_k(a)$ for $a=\frac12$ after rescaling $z\mapsto \pm\beta^{\pm\frac12}z$, the second equality holds thanks to \eqref{B identity2}, and the third equality is merely a repackaging of summation by the fact that $B_k(a)=0$ for odd $k$.
\end{proof}

\subsubsection{$\mu=-1$: refined Euler characteristic of moduli space of Klein surfaces}

 It is shown in \cite{CDO24-1,CDO24-2} that when $\mu=-1$, correlators $\omega_{g,n}$ on the Weber-type refined spectral curve generate correlators of the $\beta$-deformed Gaussian matrix model which correspond to certain weighted $b$-Hurwitz numbers \cite{CD20}. Furthermore, it is recently proved in \cite{CGGO26} that $F_g$ is related to the \emph{refined Euler characteristic} $\chi_b$ of the moduli space $\calN_{g,n}(L)$ of metric M\"{o}bius graphs, equipped with the notion of \emph{measure of non-orientability} $\rho_b:\calN_{g,n}(L)\to\BR[b]$. The parameter $b$ is related to $\beta$ by
 \begin{equation}
   \beta=\frac{1}{1+b}.
 \end{equation}

 The moduli space admits a decomposition $\calN_{g,n}(L)=\calN_{g,n}^+(L)\sqcup\calN_{g,n}^-(L)$ where $\calN_{g,n}^+$ (resp. $\calN_{g,n}^-$) denotes that underlying graphs are orientable (resp. non-orientable). In particular, when $b=0$, $\rho_{0}$ is the indicator function of orientability, i.e., $\rho_{0}=1$ on $\calN_{g,n}^+(L)$ and $\rho_{0}=0$ otherwise. When $b=1$, $\rho_{1}=1$ throughout the moduli space $\calN_{g,n}(L)$. For a generic value of $b$, the function $\rho_b$ measures how ``non-orientable'' a point in the moduli space $\calN_{g,n}(L)$ is.

 The refined Euler characteristic $\chi_b(\calN_{g,n}(L))$ of $\calN_{g,n}(L)$ is defined as the sum of $(-1)^{|E|}$ over the underlying graphs of type $(g,n)$ with $E$ edges, weighted not only by the order of the automorphism, but also by the (average of the) measure of non-orientability. Since $\calN^\pm_{g,n}(L)$ is, respectively, isomorphic to the moduli space $\calM_{g,n}$ of Riemann surfaces and the moduli space $\calK_{g,n}$ of Klein surfaces without fixed point loci (up to powers of $\BZ_2$), the refined Euler characteristic continuously interpolates Euler characteristics of two moduli spaces $\calM_{g,n}$ and $\calK_{g,n}$. See \cite{CGGO26} for more details about the interpretation of $F_g$ in this aspect.

\section{Resurgence and summability}\label{sec:resurgence}

There are two closely related notions: resurgence and Borel-summability. We begin by reviewing basic concepts, and then consider how $F^{\rm st}\in\hbar\BC[\![\hbar/t]\!]$ in ~\eqref{eq:F_stable} sits in these aspects. We refer to~\cite{dorigoni,lectures-marino} for more detailed reviews on the subject.

\subsection{Basics}

Let us denote by $\calB\colon\hbar\BC[\![\hbar]\!]\mapsto\BC[\![\zeta]\!]$ the \emph{Borel transform} that is defined as the formal inverse of the Laplace transform 
\be \label{def:laplace}
\calL^\theta[\phi](\hbar)=\int_{0}^{\re^{i\theta}\infty}\re^{-\zeta/\hbar} \phi(\zeta)\, d\zeta\,.
\ee
More precisely, 
\be\label{def:borel}
\calB\big[\hbar^{n+1}\big]\:=\frac{\zeta^n}{n!}\,,
\ee
and it extends by countably linearity over $\hbar\BC[\![\hbar]\!]$. A formal series $f\in\hbar\BC[\![\hbar]\!]$ is called \emph{Gevrey-$1$} if and only if its Borel transform $\calB[f]$ has finite radius of convergence, i.e. $\calB[f]\in\BC\{\zeta\}$.

Following Écalle's original definition \cite{EcalleI}, a germ $\phi\in\BC\{\zeta\}$ is said to admit {\em endless analytic continuation} when, for every $L > 0$, there is a finite set of points $\Omega_L$ such that $\phi$ can be analytically continued along any smooth path of length at most $L$ that avoids $\Omega_L$ and starts from a fixed point in a neighborhood of the origin. The set $\Omega=\bigcup_{L>0}\Omega_L$ is the set of singularities. A Gevrey-$1$ series $f\in\hbar\BC[\![\hbar]\!]$ is said to be \emph{resurgent} if its Borel transform $\calB[f]\in\BC\{\zeta\}$ can be endlessly analytically continued. 
If, additionally, $\calB[f]$ has only simple poles and logarithmic branch points, then $f$ is said to be {\em simple resurgent}.

If $\calB[f]$ can be analytically continued in a neigbourhood of the ray $(0,\re^{i\theta} \infty )$ and it is of exponential type $\Lambda$,\footnote{Recall that a function $\phi$ is of exponential type $\Lambda$ if for every $\varepsilon>0$, there is a constant $A_\varepsilon$ (which may depends on $\varepsilon$) such that $|\phi(\zeta)|\le A_\varepsilon \re^{(\Lambda+\varepsilon)|\zeta|}$.} then its Laplace transform $\calL^\theta[\calB[f]]$ in direction $\theta$ is analytic for $\Re(\re^{i\theta} /\hbar)>\Lambda$. Such a Gevrey-1 series $f$ is called \emph{Borel-summable} in direction $\theta$ and the outcome $s^\theta[f]\coloneqq\calL^\theta[\calB[f]]$ is called the \emph{Borel--Laplace sum of $f$}. So in particular, a resurgent series $f\in\hbar\BC[\![\hbar]\!]$ is Borel-summable in direction $\theta$, if $\calB[f]$ has no singularity in the $\theta$-direction, and it is of exponential type.

\bigskip

For a resurgent series $f\in\hbar\BC[\![\hbar]\!]$, let $\zeta_\omega$ be a singularity of $\calB[f]$, and we set $\theta_\omega=\arg\zeta_\omega$. Clearly, $s^{\theta_\omega}[f]$ is ill-defined, and thus, there can be a discontinuity between $s^{\theta_\omega\pm\varepsilon}[f]$ for a sufficiently small $\varepsilon$ --- for brevity, we suppress $\varepsilon$ from now on. We call a ray that starts at the origin and passes through $\zeta_\omega$ a \emph{Stokes ray}, and we call each such discontinuity
\be \label{def: disc}
\mathrm{disc}_{\theta}[f] \coloneqq s^{\theta+}[f] - s^{\theta-}[f]\, , 
\ee
a \emph{non-perturbative correction} to $f$. This is because such a discontinuity comes with a factor of $\re^{-\zeta_\omega/\hbar}$ due to the Laplace kernel, hence it is exponentially suppressed (or growing) and it does not appear in an asymptotic expansion. The data of the location of each singularity and its associated non-perturbative correction forms a so-called \emph{resurgent structure} of $f$\footnote{In some literature, non-perturbative corrections in a resurgent structure are defined only formally, i.e. the Borel transform $\calB[f]$ might not be of exponential type hence the Laplace transform might not be convergent.}. 

When the singularity $\zeta_{\omega}$ is a simple pole, the local expansion of the $\calB[f](\zeta)$ at $\zeta_\omega$ has the form 
\be \label{eq: Stokes0}
\calB[f](\zeta) = - \frac{S_{\omega}}{2 \pi i (\zeta - \zeta_{\omega})} + \text{regular in $\zeta-\zeta_\omega$} \, ,
\ee
where $S_{\omega} \in \BC$ is called the \emph{Stokes constant} at $\zeta_{\omega}$. A standard contour deformation argument shows that $\mathrm{disc}_{\theta}[f]$ \eqref{def: disc} is a \emph{trans-series}, whose coeffcients are the Stokes constants times the \emph{trans-monomials} $\re^{-\zeta_{\omega}/\hbar}$:
\be \label{eq: Stokes1-poles}
\mathrm{disc}_{\theta}[f](\hbar) = \sum_{\omega  \in \Omega_{\theta}} S_{\omega} \re^{-\zeta_{\omega}/\hbar}  \, ,
\ee
where $\theta$ is direction of a Stokes ray and the index $\omega \in \Omega_{\theta}$ labels the singularities $\zeta_{\omega}$ such that $\arg (\zeta_{\omega}) = \theta$.

\subsection{Resurgence}
In this section we investigate the resurgence properties of $F^{\rm st}$ in ~\eqref{eq:F_stable}. More concretely, we would determine singularities of $\calB[F^{\rm st}]$ and compute $\text{disc}_\theta[F^{\rm st}]$ in some suitable directions $\theta$. As discussed in Remark \ref{rem:t-dependence}, we set $t=1$ here and hereafter, and the $t$-dependence will be only restored in Theorem \ref{thm:resumable} and Theorem \ref{thm:QM} for completeness.

\begin{prop}\label{prop:resurgence}
The formal series $F^{\rm st}(\hbar;\beta)\in\BQ[\mu,\beta^{\pm\frac12}][\![\hbar]\!]$ has a simple resurgent structure whose singularities are all simple poles at 
\begin{equation}\label{eq:sing}
  \zeta_k^{\pm}=2\pi\ri\beta^{\pm 1/2} k\,, \quad  k\in\BZ_{\neq 0}\,,  
\end{equation}
and Stokes constants are 
\begin{align}\label{eq:stokes}
S_{\zeta_k^+}=\frac{\re^{\pi\ri (1-\mu)k}}{k}\frac{q^{(1+\mu)k/2}}{ q^k-1}\,,\quad
S_{\zeta_k^-}=-\frac{\re^{\pi\ri (1+\mu)k}}{k}\frac{\tilde{q}^{(1+\mu)k/2}}{\tilde{q}^{k}-1},
\end{align}
where $q=\re^{2\pi\ri\beta}$ and $\tilde q = \re^{-2\pi \ri \frac{1}{\beta}}$.

\end{prop}
\begin{proof}
Using \eqref{Bsign}, the Borel transform of $F^{\rm st}(\hbar;\beta)$ with respect to the variable $\hbar$ is
\begin{align}
\calB \big[F^{\rm st}\big](\zeta;\beta)=&\sum_{n\geq 1} \frac{\mathsf B_{2,n+2}\big(\frac{\mathfrak{b}}{2}(1+\mu)\big)}{(n+2)(n+1)n}\frac{\zeta^{n-1}}{(n-1)!},\\
=&\zeta^{-3}\sum_{n\geq 1} \mathsf B_{2,n+2}\big(\tfrac{\mathfrak{b}}{2}(1+\mu)\big)\frac{\zeta^{n+2}}{(n+2)!}\,
\end{align}

Then, it follows from the definition of double Bernoulli numbers in ~\eqref{def of doubleB} that
\begin{multline}\label{eq:borel_Fstable}
\calB \big[F^{\rm st}\big](\zeta;\beta)=\zeta^{-3}\Big[\frac{\zeta^2\re^{\frac{\mathfrak{b}}{2}(1+\mu)\zeta}}{(\re^{\beta^{1/2}\zeta}-1)(\re^{-\beta^{-1/2}\zeta}-1)}-\mathsf B_{2,0}\big(\tfrac{\mathfrak{b}}{2}(1+\mu)\big)\\
-\mathsf B_{2,1}\big(\tfrac{\mathfrak{b}}{2}(1+\mu)\big)\zeta-\mathsf B_{2,2}\big(\tfrac{\mathfrak{b}}{2}(1+\mu)\big)\frac{\zeta^2}{2}\Big]\,.
\end{multline}
Hence, $\calB \big[F^{\rm st}\big]$ has simple poles at $\zeta_k^\pm =2\pi\ri\beta^{\pm 1/2}\,k$ with $k\in\BZ_{\neq0}$, and no logarithmic singularities. Notice that $\zeta=0$ is not a singularity which can be checked by inspection.  

Finally, the Stokes constants are determined by taking residues at the poles, namely 
\begin{equation}
S^\pm_k=-2\pi\ri\,{\rm Res}_{\zeta=\zeta^\pm_k} \calB \big[F^{\rm st}\big](\zeta;\beta)\,,
\end{equation}  
which gives the statement. Since $\calB \big[F^{\rm st}\big]$ has no logarithmic singularities, determining Stokes constants fully describes the resurgent structure of $F^{\rm st}$.
\end{proof}

We note that when $\beta=1$, which corresponds to the unrefined limit, the ray of Borel singularities is along the imaginary axis. In the refined setting, however, the ray splits into two along $2\pi \ri \beta^{\pm1/2}$ with the opening angle of $\theta_\beta=\arg\beta$ (see Figure \ref{fig:the_borel_plane_of_f_rm_stable}), which is already observed in \cite{AMP23}. The Stokes constants $S_k^\pm$ coincide with those in~\cite{AMP23} when we set $\mu=0$ as implied by Proposition~\ref{prop:RTR-AMP}. 

\begin{figure}[ht]
\centering
\begin{tikzpicture}
\draw[thick] (-1,2)--(1,-2);
\draw[thick] (-1,-2)--(1,2);

\draw[->] (-2,0)--(2,0);
\draw[->] (0,-2.5)--(0,2.5);

\draw (63.43:0.7) arc (63.43:116.57:0.7);
\node[font=\tiny, above] at (0.2, 0.7) {$\theta_\beta$};

\foreach \x in {-0.75,-0.5,-0.25,0.25,0.5,0.75}{
        \fill (-1*\x,2*\x) circle[radius=2pt];
        \fill (-1*\x,-2*\x) circle[radius=2pt];
    }

\node[font=\tiny,below right] at (1*0.75,2*0.75) {$\zeta_k^-$};  
\node[font=\tiny,below left] at (-1*0.75,2*0.75) {$\zeta_k^+$};   
\end{tikzpicture}
\caption{The Borel plane of $F^{\rm st}(\hbar,t;\beta)$.}\label{fig:the_borel_plane_of_f_rm_stable} 
\end{figure}
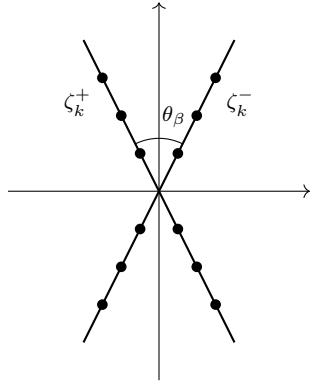

Here and hereafter, we choose $\beta\in\BH$ and the branch of $\beta^\frac12$ along the negative reals which ensures $|q|,|\tilde q|<1$ and also $\text{Re}(\beta^{\pm1/2})>0$. It turns out that we can explicitly compute the discontinuities in terms of $q$-Pochhammer symbols $(a;q)_n\coloneqq\prod_{k=0}^{n-1}(1-aq^k)$. 

\begin{lemma}\label{lem:q-pochh}
Let us denote by $\theta_\pm=\arg\ri\beta^{\pm1/2}$. The discontinuities of $F^{\rm st}$ in the directions of $\theta_\pm$ and $\theta_\pm+\pi$ are given as below: 
\begin{align}
  {\rm disc}_{\theta_+}[F^{\rm st}](\hbar;\beta)=&\log(\re^{\pi\ri(1-\mu)-2\pi\ri\beta^{1/2}/\hbar}\,q^{\frac{1+\mu}{2}};q)_\infty \,,  &\Im\left(-\mu- 2\beta^{1/2}\hbar^{-1}\right)>0\,,\label{eq:beta-pos-disc}\\
  {\rm disc}_{\theta_-}[F^{\rm st}](\hbar;\beta)=&-\log(\re^{\pi\ri(1+\mu)-2\pi\ri\beta^{-1/2}/\hbar}\,\tilde{q}^{\frac{1+\mu}{2}};\tilde{q})_\infty \,,  &\Im\left(\mu- 2\beta^{-1/2}\hbar^{-1}\right)>0\,,\\
  {\rm disc}_{\theta_++\pi}[F^{\rm st}](\hbar;\beta)=&-\log(\re^{\pi\ri(1+\mu)+2\pi\ri\beta^{1/2}/\hbar}\,q^{\frac{1-\mu}{2}};q)_\infty \,,  &\Im\left(\mu+2\beta^{1/2}\hbar^{-1}\right)>0\,,\label{eq:beta-neg-disc}\\
  {\rm disc}_{\theta_-+\pi}[F^{\rm st}](\hbar;\beta)=&\log(\re^{\pi\ri(1-\mu)+2\pi\ri\beta^{-1/2}/\hbar}\,\tilde{q}^{\frac{1-\mu}{2}};\tilde{q})_\infty \,,  &\Im\left(-\mu+2\beta^{-1/2}\hbar^{-1}\right)>0\,.
\end{align}
\end{lemma}
\begin{proof}Since arguments are similar to all four rays, we only focus on the discontinuity along the direction $\theta_+$. 

Recall we have already computed the corresponding Stokes constants $S_k^+$ in Proposition \ref{prop:resurgence}, and thus, it follows from \eqref{eq: Stokes1-poles} that
\be
\begin{aligned}
{\rm disc}_{\theta_+}[F^{\rm st}](\hbar;\beta)
&=\sum_{k\geq 1}\frac{\re^{\pi\ri(1-\mu)k}}{k}\frac{q^{\frac{(1+\mu)k}{2}}}{q^k-1}\re^{-2\pi\ri k\beta^{1/2}/\hbar}\\
&=-\sum_{k\geq 1}\frac{(\re^{\pi\ri(1-\mu)-2\pi\ri\beta^{1/2}/\hbar}\,q^{\frac{1+\mu}{2}})^k}{k}\sum_{\ell=0}^\infty q^{k\ell}\,,\\
&=\sum_{\ell=0}^\infty\log(1-(\re^{\pi\ri(1-\mu)-2\pi\ri\beta^{1/2}/\hbar}\,q^{\frac{1+\mu}{2}})q^\ell),\\
\end{aligned}
\ee 
where the summations in $k$ and $\ell$ commute because the series is absolutely convergent. Note that the absolutely convergence holds under the condition of $\Im\left(-\mu- 2\beta^{1/2}\hbar^{-1}\right)>0$. Finally, permuting the logarithm with the sum leads the desired result.  
\end{proof}
Notice that when $\mu\in\BZ$, the discontinuities at $\theta_\pm+\pi$ can be expressed in terms of the ones at $\theta_\pm$ as follows:
\begin{equation}
{\rm disc}_{\theta_\pm+\pi}[F^{\rm st}](\hbar;\beta)=-{\rm disc}_{\theta_\pm}[F^{\rm st}](-\hbar;\beta)-\text{sgn}(\mu)\log((-1)^{1+\mu} \re^{2\pi\ri\beta^{1/2}/\hbar}\,q^{\frac{1-|\mu|}{2}};q)_{|\mu|}\,.
\end{equation}
This can be derive by utilising the identities 
\begin{subequations}
\begin{align}
(t;q^{-1})_\infty &=(qt;q)_\infty^{-1}\,,\\
(wq^n;q)_\infty &=\frac{(w;q)_\infty}{(w;q)_n}\,,\quad n\in\BZ\,. \label{eq:pochh-id}
\end{align}
\end{subequations}

\subsection{Borel-summability}

We turn to the Borel-summability of $F^{\rm st}(\hbar;\beta)$ for fixed $\beta$. Recall that, for $x,a_1,a_2$ with ${\rm Re}(x),{\rm Re}(a_1),{\rm Re}(a_2)>0$, the \emph{Barnes double zeta function} may be defined as the following series
\begin{align}
  \zeta_2(s,x|\underline{a})\coloneqq\sum_{n_1,n_2\in\BZ_{\geq0}}\frac{1}{(x+n_1a_1+n_2a_2)^s},
\end{align}
which is absolutely convergent when ${\rm Re}(s)>2$. It can be analytically continued to a meromorphic function of $s\in\BC$ with simple poles at $s\in\{1,2\}$. For fixed $s\in\BC\setminus\BZ_{>0}$, it can also be analytically continued to a meromorphic function of $x\in\BC$. Then, the \emph{Barnes double gamma function} is defined by
\begin{equation}
  \Gamma_2(x|\underline{a})\coloneqq\exp\partial_s\zeta_2(s,x|\underline{a})|_{s=0},
\end{equation}
that is a meromorphic function of $x\in\BC$ whose poles are all simple located at the lattice $\Lambda=\{-(m_1a_1+m_2a_2)|m_1,m_2\in\BZ_{\geq0}\}$. We refer the readers to \cite{multiple_gamma,BBS19} for more details\footnote{The Barnes multiple zeta and gamma function can be defined for $\underline{a}=(a_1,a_2,...,a_N)$ in general.}, and in particular, one can extend the results of \cite{multiple_gamma} as below:

\begin{lemma}\label{lem:Borel_sum_BBS}
Let ${\Psi}(x,\delta;a_1,a_2)$ be the asymptotic expansion of $\log\Gamma_2(x+\delta\vert a_1, a_2)$ as $x\to\infty$, with $\Re(a_1),\Re(a_2)>0$ and $\delta\in\BC$. Then, for ${\rm Re}(x)>0$,
\be
\log\Gamma_2(x+\delta\vert a_1, a_2)=s^0[{\Psi}](x,\delta;a_1,a_2)\,.
\ee
\end{lemma}
\begin{proof}
It is well-known (c.f.~\cite[Cor. A.1]{BBS19}) that the explicit expression of ${\Psi}$ is
\be\label{eq:tilde-Phi}
{\Psi}(x,\delta;\underline{a})=-\frac{1}{2} B_{2,2}(x+\delta\vert \underline{a})\log(x)+\frac{3x^2}{4}B_{2,0}(\delta|\underline{a})+x B_{2,1}(\delta\vert \underline{a})+\sum_{k>0}(-1)^k\frac{ B_{2,k+2}(\delta\vert \underline{a})}{k(k+1)(k+2)}x^{-k}\,.
\ee
Hence, directly computing the Borel sum of ${\Psi}$ leads to (c.f. \eqref{eq:borel_Fstable})
\begin{multline}
s^0[{\Psi}](x,\delta;\underline{a})=-\frac{1}{2} B_{2,2}(x+\delta\vert \underline{a})\log(x)+\frac{3x^2}{4}B_{2,0}(\delta|\underline{a})+x B_{2,1}(\delta\vert \underline{a})\\
+\int_0^\infty \frac{d\zeta}{\zeta} \re^{-\zeta x}\bigg[\frac{\re^{-\delta\zeta}}{(\re^{-a_1\zeta}-1)(\re^{-a_2\zeta}-1)}-\frac{ B_{2,2}(\delta\vert\underline{a})}{2}+\frac{ B_{2,1}(\delta\vert\underline{a})}{\zeta}-\frac{ B_{2,0}(\delta\vert\underline{a})}{\zeta^2}\bigg]\,,
\end{multline}
which for $\delta=0$ agrees with \cite[(3.13-14)]{multiple_gamma} that give an integral representation of the double gamma function $\Gamma_2(0\vert\underline{a})$. Then, generalizing the proof of \cite[(3.13-14)]{multiple_gamma}\footnote{It is enough to replace the function $f(t)$ chosen in \cite[(3.1)]{multiple_gamma} with $f_\delta(t)=t^N \re^{-\delta t}\prod_{j=1}^N(1- \re^{-a_j t})^{-1}$. See also the discussion after the proof of \cite[Proposition 2.3]{multiple_gamma}. In addition, the proof extends to the case with $\text{Re}(a_1),\text{Re}(a_2)>0$, though $a_1,a_2$ are assumed to be real positive in \cite{multiple_gamma}.} one derive the expected result.
\end{proof} 

We can now deduce the summability of $F^{\rm st}(\hbar;\beta)$.

\begin{theorem}\label{thm:resumable}
The formal series $F^{\rm st}(\hbar;\beta)\in\BQ[\mu,\beta^{\pm1/2}][\![\hbar/t]\!]$ in ~\eqref{eq:F_stable} is Borel-summable, and for ${\rm Re}(\hbar/t)>0$, its sum along the real positive direction is given by the following expression:
\begin{equation}\label{eq:borel_sum_Fstable}
s^0[F^{\rm st}](\tfrac{\hbar}{t};\beta)+F^{\text{unst}}(\tfrac{\hbar}{t};\beta)=\log\Gamma_2(\tfrac{t}{\hbar}+\tfrac{\mathfrak{b}}{2}(1-\mu)\vert \beta^{1/2},-\beta^{-1/2})\,.
\end{equation}
\end{theorem}
\begin{proof}
Recall (c.f. \cite{Nar04}) that double Bernoulli polynomials obey
\begin{equation}
  B_{2,n}(x\vert a_1,-a_2)=- B_{2,n}(x+a_2\vert a_1,a_2) \label{id:bernoulli} \,,
\end{equation}
or equivalently for the double Gamma function 
\begin{equation}
  \Gamma_2(x|a_1,-a_2)\Gamma_2(x+a_2|a_1,a_2)=1.\label{id:Gamma}
\end{equation}
Then, applying the identity~\eqref{id:bernoulli} to ~\eqref{eq:tilde-Phi} and recalling the definitions of $F^{\rm unst}$ \eqref{F^unst} and $F^{\rm st}$ \eqref{eq:F_stable}, we find that 
\be
{\Psi}(\tfrac{1}{\hbar},\tfrac{\mathfrak b}{2}(1-\mu)+\beta^{-1/2}\vert \beta^{1/2} ,\beta^{-1/2})=-F^{\text{unst}}(\hbar,\beta)-F^{\rm st}(\hbar;\beta)\,,
\ee
This proves the Borel--Laplace summability. 
The final expression ~\eqref{eq:borel_sum_Fstable} follows from \eqref{id:Gamma} and Lemma~\ref{lem:Borel_sum_BBS}. Finally, the $t$-dependence is recovered by rescaling.
\end{proof}

\subsubsection{Borel-summability in $\beta$}\label{sec:borel_summability_in_beta_}

Since the Borel sum $s^\theta[F^{\text{st}}]$ depends on both $\hbar$ and $\beta$, one may wonder the summability in terms of $\beta$ for fixed $\hbar$. In order not to digress from our main focus, however, we present a simple observation about the $\beta$-asymptotics of $\text{disc}_{\theta_\pm}[F^{\text{st}}]$ instead, and leave a thorough analysis to future work. 

Let $\psi\in\BC[\![\beta]\!]$ be a resurgent series and choose $\theta=\arg\eta_\omega$, where $\eta_\omega$ is a singular point of $\calB[\psi](\eta)$.\footnote{We use the variable $\eta$ as conjugate to the variable $\beta$ under the Borel/Laplace transforms.} Then, if $\psi$ is Borel--Laplce summable, the median resummation is defined by
\be
\calS^{\rm med}_{\theta}[\psi](\beta)\coloneqq\frac{s^{\theta+\varepsilon}[\psi](\beta)+s^{\theta-\varepsilon}[\psi](\beta)}{2}\,,
\ee
for a sufficiently small positive $\varepsilon$. Now, recall that $\epsilon_1=\beta^{1/2}\hbar$ and $\epsilon_2=-\beta^{-1/2}\hbar$, as in~\eqref{eq:epsilon}. It turns out that if we regard $(\epsilon_2,\beta)$ as independent variables, instead of $(\hbar,\beta)$, the $\beta$-asymptotics of ${\rm disc}_{\theta_\pm}[F^{\rm st}]$ enjoys an interesting property in terms of the median resummation as below:
\begin{lemma}
Set $\mu=-1$ and assume $|\Re(\epsilon_2)|<1/2$. The discontinuity ${\rm disc}_{\theta_+}[F^{\rm st}](-\beta^{\frac12}\epsilon_2;\beta)$ is reconstructed via the median resummation of its asymptotic expansion $\psi$ as $\beta\to 0$ with $\Im(\beta)>0$. More concretely, 
\begin{equation}\label{eq:s_med_f_pm}
\calS^{\rm med}_{\theta_+}[\psi](\beta)={\rm disc}_{\theta_+}[F^{\rm st}](-\beta^{\frac12}\epsilon_2;\beta)\,,
\end{equation}
where $\beta\in\BH\cap\{\Re(\ri\epsilon_2^{-1}\beta)>0\}$. In particular, 
\be\label{eq:s_0_f_pm}
s^0[\psi](\beta)={\rm disc}_{\theta_+}[F^{\rm st}](-\beta^{\frac12}\epsilon_2;\beta)+{\rm disc}_{\theta_-}[F^{\rm st}](-\beta^{\frac12}\epsilon_2;\beta)\,.
\ee
\end{lemma} 
\begin{proof}
The results follows from the summability of the Faddeev's quantum dilogarithm~\cite[Theorem~1.3]{GK} and~\cite[Theorem~A-29]{wheeler-thesis}. See also~\cite[Section~3.1.1]{FR-pmn} for similar computations.
\end{proof}

\begin{remark}
As \eqref{eq:s_med_f_pm} and \eqref{eq:s_0_f_pm} imply, ${\rm disc}_{\theta_-}[F^{\rm st}](-\beta^{\frac12}\epsilon_2;\beta)$ is the non-perturbative correction to the asymptotic expansion $\psi$ of ${\rm disc}_{\theta_+}[F^{\rm st}](-\beta^{\frac12}\epsilon_2;\beta)$ as $\beta\to0$. It is natural to ask whether the non-perturbative correction to the asymptotic expansion of ${\rm disc}_{\theta_-}[F^{\rm st}](-\beta^{\frac12}\epsilon_2;\beta)$ as $\beta\to\ri\infty$ is given by ${\rm disc}_{\theta_+}[F^{\rm st}](-\beta^{\frac12}\epsilon_2;\beta)$. This is within the spirit of the  strong-weak symmetry discussed in~\cite{FR1phys}, indeed for $\mu=-1$ the proof follows by simply applying the $S$ transformation to~\eqref{eq:s_0_f_pm}, however a general proof for all $\mu$ will require a separate analysis that is beyond the scope of this project.
\end{remark}

\section{Quantum Modularity}\label{sec:modularity}

In this section, we first review basic notions of quantum modular forms --- see \cite{Wheeler-3manifold} and references therein for more detailed discussions. Then, we prove explicit instances of quantum modularity in refined topological recursion. Our observation suggests that some of the properties may hold in a more general setting.

\subsection{Quantum modular forms}
Recall that for $k\in\frac{1}{2}\BZ$, a weight $k$ holomorphic modular form for a group $\Gamma\subseteq{\rm SL}_2(\BZ)$ is a map $f\colon\BH\to\BC$ such that 
\be
\quad \forall \gamma=\begin{pmatrix} a& b \\ c & d\end{pmatrix}\in\Gamma,\quad f\vert_{k,\gamma}(\tau)-f(\tau)=0\,,
\ee
where $\tau\in\BH$ and 
\begin{equation}
  f\vert_{k,\gamma}(\tau):=(c\tau+d)^{-k} f\left(\frac{a\tau+b}{c\tau+d}\right)\,,
\end{equation}
and with bounded growth at infinity.

Following the original proposal by Zagier~\cite{zagier_modular}, a map $f\colon\BQ\to\BC$ is then called a weight $k\in\frac{1}{2}\BZ$ quantum modular form with respect to a subgroup $\Gamma\subseteq\mathrm{SL}_2(\BZ)$ if the \emph{cocycle} $h\colon\Gamma\times\BQ\to\BC$ 
\begin{equation}\label{cocycle}
     h_\gamma[f](\tau):=f\vert_{k,\gamma}(\tau) - f(\tau) 
\end{equation} 
has ``better analyticity properties'' than the function $f$ itself for every $\gamma=\left( \begin{smallmatrix}
    a & b\\
    c & d
\end{smallmatrix} \right)\in\Gamma$---\emph{e.g.}, it is real analytic over $\BR\setminus\{\gamma^{-1}(\infty)\}$. Our focus is on so-called holomorphic quantum modular forms whose cocycle has a specific type of analytic properties: 
\begin{definition} \label{def: holoQM}
    A holomorphic function $f\colon\BH\to\BC$ is a weight $k\in\frac{1}{2}\BZ$ \emph{holomorphic quantum modular form} for a subgroup $\Gamma\subseteq\mathrm{SL}_2(\BZ)$ if its cocycle $h_\gamma[f]\colon\BH\to\BC$ extends analytically to
    \be\label{eq:cocycle-gen}
    \BC_\gamma:=\{\tau\in\BC| (c\tau+d)\in\BC\setminus\BR_{\leq 0}\}
    \ee
    for every $\gamma=\left( \begin{smallmatrix}
    a & b\\
    c & d
\end{smallmatrix} \right) \in\Gamma$.\footnote{Note that $\BC_\gamma = \BC \setminus \left(-\infty; \, -d/c\right]$ when $c>0$ and $\BC_\gamma = \BC \setminus \left[-d/c; \, +\infty \right)$ when $c<0$.}  (Holomorphic) quantum modular forms of weight zero are also known as \emph{(holomorphic) quantum modular functions}. 
\end{definition}

In a nutshell, the cocycle of a quantum modular form $f$ measures the failure of invariance under the action of the modular group $\Gamma$; clearly, the cocycle vanishes for every $\gamma\in\Gamma$ if and only if $f$ is a modular form. Hence, (holomorphic) quantum modular forms are characterized---not uniquely\footnote{Two quantum modular forms differing by a modular form of the same weight will have the same cocycle.}---by their cocycle.

In addition, as in the classical case, there is a notion of quantum Jacobi modular forms that was first introduced in~\cite{Folsom-Brig}. 
\begin{definition} \label{def: Jacobi_holoQM}
    A map $f\colon\BC\times\BH\to\BC$ is called a weight $k\in\frac{1}{2}\BZ$ and index $m\in\frac{1}{2}\BZ$ \emph{holomorphic Jacobi quantum modular form} for a subgroup $\Gamma\subseteq\mathrm{SL}_2(\BZ)$ if the following functions 
    \begin{equation}
    \begin{aligned}
        h_\gamma[f](z;\tau)&:=(c\tau+d)^{-k}\, \re^{-\frac{2\pi\ri m cz^2}{c\tau+d}} f\left(\frac{z}{c\tau+d};\frac{a\tau+b}{c\tau+d}\right) - f(z;\tau) \\
        g_{(\lambda,\nu)}[f](z;\tau)&:= \re^{2\pi\ri m (\lambda^2 \tau+2\lambda z)} f(z+\lambda \tau+\nu;\tau) - f(z;\tau)
        \end{aligned}
    \end{equation}
    extend analytically to $\BC_\gamma$
    for every $\gamma=\left( \begin{smallmatrix}
    a & b\\
    c & d
\end{smallmatrix} \right) \in\Gamma$ and every $\lambda,\nu\in\BZ$. 
The function $h:\Gamma\times\BC\times\BH\to\BC$ is still called the cocycle. In addition, the variables $\tau$ and $z$ are usually referred to as the modular and elliptic parameter, respectively.  
\end{definition}

As mentioned above, modular forms are all quantum modular because their cocycles vanish. One can  easily see that quasi-modular forms are also quantum modular. Since the set of all quasi-modular forms is generated by Eisenstein series $G_k$ of weight $k=2,4,6$ and $G_4,G_6$ are modular, it is sufficient to verify that $G_2$ is quantum modular. Recall that the Eisenstein series of weight $k$ as a $q$-series is given by
\be\label{eq:Eisenstein}
G_{k}(\tau)=-\frac{B_k}{2k}+\sum_{n=1}^\infty n^{k-1}\frac{q^{n}}{1-q^n}\,, \quad q=\re^{2\pi i\tau}\,.
\ee
The cocycle of of $G_2$ under $S=\left( \begin{smallmatrix}
    0 & -1\\
    1 & 0
\end{smallmatrix} \right) :\tau\mapsto-\tau^{-1}$ is expressed by the following equation
\begin{equation}
\tau^{-2}G_2\Big(-1/\tau\Big)-G_2(\tau)=\frac{i}{4\pi \tau}\,,
\end{equation}  
that is a meromorphic function with a simple pole at $\tau=0$, well-defined in a domain even larger than the cut plane $\BC_S=\BC\setminus\BR_{\leq 0}$, hence, meeting the requirement in Definition ~\ref{def: Jacobi_holoQM}. 

Elementary examples of quantum modular forms that are not modular, quasi-modular, nor Mock modular are Eisenstein series of odd weight \cite{zagier-talk}. Note that $G_k$ defined by \eqref{eq:Eisenstein} remains nonzero when $k$ is odd, unlike the definition as a lattice point sum. The cocycle under the $S$-transformation yields:
\begin{equation}
  h_S[G_{2k-1}](\tau)=-2\sum_{m>0}\frac{1}{m^{2k-1}}-2\sum_{n,m>0}\frac{1}{(m\tau+n)^{2k-1}}\,.
\end{equation}
The right-hand side is holomorphic on $\BC\setminus \BR_{\leq0}$, hence meeting the condition in Definition ~\ref{def: Jacobi_holoQM}. It turns out that Eisenstein series have a somewhat surprising relation to the free energy $F_g$ as shown in Remark~\ref{rmk:large hbar}. 

\subsection{Faddeev's quantum dilogarithm}\label{sub:faddeev_s_quantum_dilogarithm}

Let $\epsilon>0$ be a sufficiently small positive real number. For $x\in\BC$ and $\bb\in\BC\setminus \ri \BR$ satisfying $\text{Im}(2x)<\text{Re}(\bb+\bb^{-1})$, \emph{Faddeev's quantum dilogarithm} \cite{Faddeev1} is defined by
\be \label{eq: seriesPhib}
\Phi_{\bb}(x) = \exp \left( \int_{\BR + \ri \epsilon} \frac{\re^{-2 \ri x z}}{4 \sinh(z \bb ) \sinh(z \bb^{-1})} \frac{d z}{z} \right).
\ee
Clearly, it satisfies 
\be \label{eq: symmPhib}
\Phi_{\bb}(x) = \Phi_{-\bb}(x) = \Phi_{\bb^{-1}}(x) \, .
\ee 
Let us set $\cb = \ri (\bb + \bb^{-1})/2 \, ,\quad q = \re^{2 \pi \ri \bb^2} \, , \quad \tilde{q} = \re^{- 2 \pi \ri \bb^{-2}} \, $. In the domain $\text{Im}(\bb^2)>0$ (equivalently $\bb$ is inside the first and third quadrant), we have $|q|,|\tilde q|<1$ and thus $\Phi_{\bb}(x)$ can be written in terms of $q$-Pochhammer symbols as
\begin{equation}\label{Faddeev and q-Pochhammer}
  \Phi_{\bb}(x)=\frac{( \re^{2 \pi \bb (x + \cb)}; \, q)_{\infty}}{(  \re^{2 \pi \bb^{-1} (x - \cb)}; \, \tilde{q})_{\infty}}.
\end{equation}
This expression can be extended to the domain $\text{Im}(\bb^2)<0$ (equivalently $\bb$ is inside the second and fourth quadrant) thanks to the identities \eqref{eq: symmPhib}, and in fact, one can also show that it is well-defined in $\bb\in\BC\setminus \ri \BR$ as we initially stated. Then, \eqref{Faddeev and q-Pochhammer} immediately shows that $ \Phi_{\bb}(x)$ has simple zeroes and poles at
\be
x = -c_{\bb} - \ri m \bb - \ri n \bb^{-1} \, ,\quad x = c_{\bb} + \ri m \bb + \ri n \bb^{-1} \, ,
\ee
for $m,n \in \BN$, respectively. 

For fixed $x$, the asymptotic expansion of Faddeev's quantum dilogarithm as $\bb^2\to0$ is
\be \label{eq: logPhib}
\log \Phi_{\bb}\left( \frac{x}{2 \pi \bb} \right)\sim \sum_{k=0}^{\infty} (2 \pi \ri \bb^2)^{2k-1} \frac{B_{2k}(1/2)}{(2k)!} \mathrm{Li}_{2-2k}(-\re^x) \, ,
\ee
where $\mathrm{Li}_n(z)$ is the polylogarithm of order $n$, and $B_n(z)$ is the $n$-th Bernoulli polynomial. It is a divergent series, but Borel--Laplace summable, and it indeed coincides with its Borel--Laplace sum in the positive real direction \cite[Theorem 1.3]{GK}. In other words, $\Phi_\bb(x)$ is {\em Borel regular} as defined in~\cite{borel_reg}. 

\subsubsection{Applications}

Faddeev's quantum dilogarithm $\Phi_\bb(x)$ plays important roles in a variety of subjects in mathematics and physics. Historically, it is known to be one of the building block of quantum invariants of three-manifolds. More recently it also appeared in the study of the so-called fermionic traces~\cite{KM15}. 

In the study of quantum modular forms, a new class of $q$-series, called {\em modular resurgent}, has been recently defined in terms of the resurgent structure of their asymptotic expansion, by the first author and C. Rella~\cite[Definition~3.2]{FR1maths}. These $q$-series have been also conjectured to give rise to quantum modular forms~\cite[Conjecture~4]{FR1maths}. This conjecture is supported by a number of evidence \cite[Theorem~3.10]{FR-pmn}, \cite[Theorems~4.11,~4.12]{FR-pmn} whose proofs utilise features of $\Phi_\bb(x)$, e.g. the Borel--Laplace summability and its relation to $q$-Pochhammer symbols \eqref{Faddeev and q-Pochhammer}. 
\subsection{Relation to refined topological recursion}\label{sec:QM}

We now discuss our main observation, namely, quantum modularity in refined topological recursion. A crucial point is that we should treat $(\epsilon_2,\beta)$ as independent variables, instead of $(\hbar,\beta)$. Before turning to our main theorem, let us justify that why such a choice is not unnatural.

First of all, we have chosen $(\beta,\epsilon_2)$ for notational convenience, but one can equivalently take $(-\beta^{-1},\epsilon_1)$ as independent variables, and all our results hold (with tiny adjustments of signs). It is then worth noting that, as discussed in \cite{CDO24-2,CGGO26}, combinatorial and geometric interpretations arise after rescaling $\omega_{g,n}$ by $\beta^{-g/2}$ which is equivalent to considering the $\epsilon_1$-expansion instead of $\hbar$. In summary, although coefficients in the $\hbar$-expansion naturally appear in refined topological recursion, it has been observed that they admit better properties after rescaling by certain powers of $\beta$. It turns out that modularity also goes along this line.

With this setting, we prove our main theorem stating that the discontinuities in Lemma \ref{lem:q-pochh} as functions of $(\epsilon_2,\beta)$ are Jacobi quantum modular functions.

\begin{theorem}\label{thm:QM}
Set $\mu\in\BZ$. The discontinuities ${\rm disc}_{\theta_\pm}[F^{\rm st}](-\beta^{\frac12}\epsilon_2/t;\beta)$ are Jacobi quantum modular functions with respect to the elliptic variable $t/\epsilon_2$ and the modular parameter~$\beta$.
\end{theorem}
\begin{proof}
As discussed in Remark \ref{rem:t-dependence}, the $t$-dependence is recovered by rescaling $\hbar$. To lighten the notation, we set in this proof $u_+=-\tfrac{1+\mu}{2}+\tfrac{1}{\epsilon_2}$, $u_-=\tfrac{1+\mu}{2}+\tfrac{1}{\beta\epsilon_2}$, and also we denote by
\begin{equation}
  f_{\pm}(\epsilon_2^{-1};\beta)=\
{\rm disc}_{\theta_\pm}[F^{\rm st}](-\beta^{\frac12}\epsilon_2;\beta)
\end{equation}
Recall that it is enough to prove the statements for the generators of ${\rm SL}_2(\BZ)$. 

First notice that the $T$-cocycle $h_T[f_{+}](\epsilon_2^{-1};\beta)$ is trivial as it is written as a $q$-series. Next, the $S$-cocycle $h_S[f_{+}](\epsilon_2^{-1};\beta)$ can be computed by manipulation of \eqref{eq: seriesPhib} and the identity \eqref{eq:pochh-id}. Explicitly, we have
\be\label{eq:S-cocycle-mu0}
\log(\re^{2\pi i(-\frac{1+\mu}{2}+\tfrac{1}{\beta\epsilon_2})}\tilde{q}^{\frac{1+\mu}{2}};\tilde{q})_\infty-\log(\re^{2\pi \ri u_+} q^{\frac{1+\mu}{2}};q)_\infty=\log\frac{\Phi_{\beta^{1/2}}\big(\frac{\ri}{\epsilon_2\beta^{1/2}}-\frac{\ri\mu}{2}\big)}{(\re^{2\pi iu_+};q)_{-\mu}}\,,
\ee
which holds when $\mu\in\BZ$. We then deduce that 
\be\label{eq:coc_S}
h_S[f_{+}](\epsilon_2^{-1};\beta) =\log\Phi_{\beta^{1/2}}\Big(\frac{\ri}{\epsilon_2\beta^{1/2}}-\frac{\ri\mu}{2}\Big)-\log(\re^{2\pi \ri u_+};q)_{-\mu}\,,
\ee
which is analytic for $\beta\in\BC\setminus\BR_{\le 0}$ and $\frac{\ri}{\epsilon_2\beta^{1/2}}=-\ri\frac{t}{\hbar}\in\BC$.\footnote{When $\mu=-1$, the result follows directly from~\cite[Theorem~A-41]{wheeler-thesis}.}

A similar argument works for $f_-(\epsilon_2^{-1};\beta)$. Indeed, for $\mu\in\BZ$, the $S$-action transforms ${\rm disc}_{\theta_+}$ into ${\rm disc}_{\theta_-}$ with sign, namely
\begin{equation}
{\rm disc}_{\theta_+}[F^{\rm st}](\epsilon_2^{-1};\beta)\,\big\vert_S=\log(\re^{2\pi\ri/\beta\epsilon_2};\tilde{q})_\infty=-{\rm disc}_{\theta_-}[F^{\rm st}](\epsilon_2^{-1};\beta)\,,
\end{equation}
ensuring that the $S$-cocycle of $f_-(\epsilon_2^{-1};\beta)$ is given by~\eqref{eq:coc_S}.

Finally, we verify another constraint, i.e. the condition on the function $g_{(\lambda,\mu)}$ in Definition~\ref{def: Jacobi_holoQM} with $\lambda,\nu\in\BZ$: for $f_{+}(\epsilon_2^{-1};\beta)$, we get 
\be
\begin{aligned}
g_{(\lambda,\nu)}[f_{+}](\epsilon_2^{-1};\beta)&=\log(\re^{2\pi i(-\frac{\mu+1}{2}+\frac{1}{\epsilon_2}+\lambda \beta+\nu)}q^{\frac{\mu+1}{2}};q)_\infty-\log(\re^{2\pi i(-\frac{\mu+1}{2}+\frac{1}{\epsilon_2})}q^{\frac{\mu+1}{2}};q)_\infty\\
&=\log(\re^{2\pi iu_+}q^{\frac{1+\mu}{2}};q)_\lambda\,,
\end{aligned}
\ee
where we used the identity~\eqref{eq:pochh-id}. Similarly, for $f_{-}$ we find
\be
\begin{aligned}
g_{(\lambda,\nu)}[f_{-}](\epsilon_2^{-1};-\beta^{-1})&=-\log(\re^{2\pi i(\frac{\mu+1}{2}+\frac{1}{\epsilon_2\beta}-\frac{\lambda\beta+\nu}{\beta})}\tilde{q}^{\frac{\mu+1}{2}};\tilde{q})_\infty+\log(\re^{2\pi i(\frac{\mu+1}{2}+\frac{1}{\epsilon_2\beta})}\tilde{q}^{\frac{\mu+1}{2}};\tilde{q})_\infty\\
&=-\log(\re^{2\pi iu_-}\tilde{q}^{\frac{1+\mu}{2}};\tilde{q})_\nu\,.
\end{aligned}
\ee
\end{proof}

One can similarly prove quantum modularity of ${\rm disc}_{\theta_\pm+\pi}$.

\begin{remark}\label{rmk:large hbar}
Interestingly, when $\mu=1$, the asymptotic expansion\footnote{In fact, \cite[Corollary~8]{wheeler-thesis} implies that ${\rm disc}_{\theta_+}[F^{\rm st}](-\beta^{1/2}\epsilon_2;\beta)=\sum_{n=0}^\infty\frac{1}{n!}\sum_{k\geq 1}\frac{k^n}{k}\frac{q^k}{q^{k}-1} \left(2\pi\ri\epsilon_2^{-1}\right)^n$ is convergent.} of the discontinuity ${\rm disc}_{\theta_+}[F^{\rm st}]$ as $\epsilon_2\to\infty$ and $\beta$ fixed becomes the generating function of the Eisenstein series of all weight. In fact, for $\mu=1$, we have
\be\label{eq:eis_disc}
\begin{aligned}
{\rm disc}_{\theta_+}[F^{\rm st}](-\beta^{1/2}\epsilon_2;\beta)&\sim\sum_{n=0}^\infty\frac{1}{n!}\sum_{k\geq 1}\frac{k^n}{k}\frac{q^{k}}{q^k-1} \left(\frac{2\pi\ri}{\epsilon_2}\right)^n\\
&=\log(q;q)_\infty+\sum_{n=1}^\infty\frac{1}{n!}\frac{B_n}{2n}\left(\frac{2\pi\ri}{\epsilon_2}\right)^{n}-\sum_{n=1}^\infty\frac{G_n(\beta)}{n!}\left(\frac{2\pi\ri}{\epsilon_2}\right)^{n}\,.
\end{aligned}
\ee
\end{remark}

\subsubsection{Concluding comments and future directions}\label{sec:conclusion}
The $S$-cocycle of ${\rm disc}_{\theta_+}[F^{\rm st}](-\beta^{1/2}\epsilon_2;\beta)$ corresponds to the sum of ${\rm disc}_{\theta_\pm}[F^{\rm st}]$ which can also be written as
\begin{equation}
h_S[f_{+}](\epsilon_2^{-1};\beta)=s^{\theta_--\varepsilon}[F^{\rm st}](\hbar;\beta)-s^{\theta_++\varepsilon}[F^{\rm st}](\hbar;\beta)\,
\end{equation}
for a sufficiently small $\varepsilon\in\mathbb{R}_{>0}$. This observation leads us to the following structural question in the refined setting:
\begin{question}
Let $F_g$ (resp. $F_g^\beta$) be the genus-$g$ free energy of (resp. refined) topological recursion on a certain (resp. refined) spectral curve. Define their generating $\hbar$-series $F^{\text{st}}$ (resp. $F^{\text{st},\beta}$) as in \eqref{eq:F_stable}. Now, let us assume that $\calB[F^{\text{st}}]$ has a sequence of singularities at $\zeta\in\zeta_\omega\BZ_{\neq 0}$, and also assume that there may exist the associated parameter $\mu_\omega$ in $F^{\text{st},\beta}$. Then, three natural questions arise from Theorem \ref{thm:QM}:
\begin{description}
  \item[Q1] Does $\calB[F^{\text{st},\beta}]$ have a pair of singularities at $\zeta\in\zeta_\omega\beta^{\pm1/2}\BZ_{\neq0}$?
  \item[Q2] When $\mu_\omega\in\BZ$, is the discontinuity at each ray a Jacobi quantum modular form?
  \item[Q3] When $\mu_\omega\in\BZ$, does the sum of the pair correspond to the $S$-cocycle?
\end{description}
\end{question}
These questions are well-supported by the set of examples discussed in \cite{KO22,KO23} where free energies are proved/conjectured in terms of double Bernoulli polynomials. Since our proof of Theorem \ref{thm:QM} relies on explicit expressions in terms of $q$-Pochhammer symbols, however, a conceptual understanding of quantum modularity with respect to $\beta$ is still missing. We hope that exploring these questions for nontrivial curves (e.g. higher-genus or singular refined spectral curves) would help unveiling a more fundamental reason of this feature.

Although the current formalism of refined topological recursion \cite{KO22,Osu23-1} cannot be applied to mirror curves of toric Calabi--Yau threefolds, we note that the free energy for the resolved conifold is computed in \cite{HK10} (see \cite[Appendix A]{BS24} for an explicit formula for the genus-$g$ free energy where their $(\epsilon_1,\epsilon_2)$ should be set to $(\beta^{1/2},- \beta^{-1/2})$)\footnote{We acknowledge T. Bridgeland and A. Brini for suggesting to consider the resolved conifold.}
\begin{equation}
 F_g^{\rm res. coni.}(t_r)= c_g(\beta)\sum_{n\in\mathbb{Z}}\frac{1}{(t_r+2\pi i n)^{2g-2}}\,,\quad g\geq2
\end{equation}
where $c_g(\beta)$ is given in \eqref{cg(beta)}, $t_r$ is the corresponding K\"{a}hler parameter and we have expanded ${\rm Li}_{3-2g}(e^{-t_r})$ in \cite{BS24}. Thus, for the resolved conifold, there are infinitely many rays labeled by $t_r+2\pi i n$, and \textbf{Q1-Q3} stand correct for each ray (though necessarily $\mu=0$ as discussed in Sec.~\ref{ssub:mu_0}). More generally, \textbf{Q1-Q3} are consistent with the conifold gap argument \cite{AMP23}.

\bigskip

On a related but different note, it has been shown in \cite{IK21,IK20} that topological recursion on a degree-two genus-zero spectral curve gives a solution to the Bridgeland's Riemann--Hilbert problem \cite{Bri16}. One may wonder if the story can be lifted to the refined setting, i.e. a \emph{quantum Riemann--Hilbert problem} due to \cite{BBS19}. More concretely, one may explore a relation between Stokes constants and the BPS invariants of~\cite{BBS19,KO23,KW24} for several genus-zero refined spectral curves including the Weber-type, in line with the discussions in \cite{AMP23,Chu25}. It turns out that this idea contains several technical subtleties, and the details will be reported soon by the second author of the present paper and O. Kidwai.

\newpage
\printbibliography
\end{document}